\documentclass[a4paper,11pt]{article}

\usepackage{a4wide}
\usepackage{amsmath,amssymb,amsthm}
\usepackage{array}
\newcolumntype{L}{>{\displaystyle}l}
\newcolumntype{C}{>{\displaystyle}c}
\newcolumntype{R}{>{\displaystyle}r}

\newtheorem{theorem}{Theorem}
\newtheorem{lemma}[theorem]{Lemma}
\newtheorem{corollary}[theorem]{Corollary}

\theoremstyle{definition}
\newtheorem{algorithm}{Algorithm}

\begin{document}

\title{Arboricity, $h$-Index, and Dynamic Algorithms}

\author{%
  Min Chih Lin~\thanks{Universidad de Buenos Aires, Facultad de Ciencias Exactas y Naturales, Departamento  de  Computaci\'on, Buenos Aires, Argentina. {\small \texttt{\{oscarlin, fsoulign\}.dc.uba.ar}}} \and%
  Francisco J.\ Soulignac~$^\ast$ \and%
  Jayme L.\ Szwarcfiter~\thanks{Universidade Federal do Rio de Janeiro, Instituto de Matem\'atica, NCE and COPPE, Caixa Postal 2324, 20001-970 Rio de Janeiro, RJ, Brasil. {\small \texttt{jayme@nce.ufrj.br}}}%
}

\date{}

\maketitle   

\begin{abstract}
In this paper we present a modification of a technique by Chiba and Nishizeki [Chiba and Nishizeki: Arboricity and Subgraph Listing Algorithms, SIAM J.\ Comput.\ 14(1), pp.~210--223 (1985)]. Based on it, we design a data structure suitable for dynamic graph algorithms. We employ the data structure to formulate new algorithms for several problems, including  counting subgraphs of four vertices, recognition of diamond-free graphs, cop-win graphs and strongly chordal graphs, among others. We improve the time complexity for graphs with low arboricity or $h$-index.

\vspace*{3mm} {\bf Keywords:} data structures, dynamic algorithms, arboricity, $h$-index, cop-win graphs, diamond-free graphs, strongly chordal graphs.
\end{abstract}

\section{Introduction}

We describe a variation of a technique by Chiba and Nishizeki~\cite{ChibaNishizekiSJC1985}, leading to a data structure for graph algorithmic problems, called the \emph{$h$-graph} data structure. It supports operations of insertion and removal of vertices, as well as insertion and removal of edges. Although the data structure can be used for general purpose, it is particularly suitable for applications in dynamic graph algorithms.

As an application of this data structure, we describe new algorithms for several graph problems, as listing cliques; counting subgraphs of size 4; recognition of diamond-free graphs; finding simple, simplicial and dominated vertices; recognition of cop-win graphs; recognition of strongly chordal graphs. We remark that no previous such dynamic algorithms exist so far in the literature. On the other hand, in some cases, there is also an improvement in time complexity (relative to static graph algorithms) for graphs of low arboricity or $h$-index. 

A dynamic data structure designed for graphs with low $h$-index has been first defined by Eppstein and Spiro~\cite{EppsteinSpiro2009}.  This data structure keeps, for each graph $G$ with $h$-index $h$, the set of vertices with degree at least $h$, and a dictionary that indicates the number of two-edges paths between any pair of vertices, for all those vertices at distance $2$.  The total size of the data structure is $O(mh)$ bits.  Using this information, the authors show how to maintain the family of triangles of $G$ in $O(h)$ randomized amortized time while edges are inserted or removed.  The authors also show how to keep other statistics of $G$ with this data structure.  Our $h$-graph data structure follows a different approach.  First, we store no more than the adjacency lists of $G$ in a special format, using $O(n+m)$ bits.  Second, we do not compute the $h$-index of $G$.  The advantage of our data structure is that we can use to maintain the family of triangles of $G$ in $O(dh)$ deterministic worst case time per vertex insertion or removal, where $d$ is the degree of the vertex.  Furthermore, the time required by this algorithm when it is applied to all the vertices of the graph, so as to compute the family of triangles of $G$, is $O(\alpha m)$, where $\alpha \leq h$ is the arboricity of $G$.  The disadvantage is that we cannot longer maintain the triangles as efficiently as Eppstein and Spiro when edge operations are allowed.  So, though both data structures have some similarities in their inceptions, they are better suited for different applications.  In particular, our $h$-graph data structure allows the ``efficient'' examination of the subgraph of $G$ induced by the neighborhood of an inserted or removed vertex.

One of the similarities between the $h$-graph data structure and the data structure by Eppstein and Spiro, is that both differentiate between low and high degree vertices.  The technique of handling differently vertices of high and low degree has been first employed by Alon et al.~\cite{AlonYusterZwickA1997}, and since then many other works made use of this classification (e.g.~\cite{EppsteinSpiro2009,KloksKratschMullerIPL2000}).  However, we use a local classification on each vertex.  So, some vertices can be considered as both high and low depending on the local classification of each of its neighbors.

The paper is organized as follows.  In the next section we introduce the notation and terminology employed.  In Section~\ref{sec:chiba-nishizeki} we discuss the technique by Chiba and Nishizeki, and its variation for dynamic graphs.  The $h$-graph data structure is described in Section~\ref{sec:data structure}, together with the operations that it supports.  Finally, in Section~\ref{sec:applications}, we show the applicability of the $h$-graph data structure by solving the problems listed above.  

\section{Preliminaries}

In this paper we work with undirected simple graphs.  Let $G$ be a graph with vertex set $V(G)$ and edge set $E(G)$, and call $n = |V(G)|$ and $m = |E(G)|$.  Write $vw$ to denote the edge of $G$ formed by vertices $v, w \in V(G)$.  For $v \in V(G)$, represent by $N_G(v)$ the subset of vertices adjacent to $v$, and let $N_G[v] = N_G(v) \cup \{v\}$. The set $N_G(v)$ is called the \emph{neighborhood} of $v$, while $N_G[v]$ is the \emph{closed neighborhood} of $v$.  The \emph{edge-neighborhood} of $v$, denoted by $N_G'(v)$, is the set of edges whose both endpoints are adjacent to $v$.  Similarly, the \emph{neighborhood} $N_G(vw)$ of an edge $vw$ is the set of vertices that are adjacent to both $v$ and $w$.  All the vertices in $N_G(vw)$ are said to be \emph{edge-adjacent} to $vw$.  The \emph{degree} of $v$ is $d_G(v) = |N(v)|$, the \emph{degree} of $vw$ is $d_G(vw) = |N_G(vw)|$, and the \emph{edge-degree} of $v$ is $d_G'(v) = |N_G'(v)|$.  When there is no ambiguity, we may omit the subscripts from $N$ and $d$.

For $W \subseteq V(G)$, denote by $G[W]$ the subgraph of $G$ induced by $W$, and write $E_G(W)$ to represent $E(G[W])$.  As before, we omit the subscript when there is no ambiguity about $G$.  A \emph{clique} is a set of pairwise adjacent vertices.  We also use the term \emph{clique} to refer to the corresponding induced subgraph.  The clique of size $k$ is represented by $K_k$, and the graph $K_3$ is called a \emph{triangle}.  We shall denote by $O(n^\omega)$ the time required for the multiplication of two $n\times n$ matrices.  Up to this date, the best bounds on $n^\omega$ are $n^2 \leq n^\omega < n^{2.376}$~\cite{CoppersmithWinogradJSC1990}.  The \emph{arboricity} $\alpha(G)$ of $G$ is the minimum number of edge-disjoint spanning forests into which $G$ can be decomposed.  The \emph{$h$-index} $h(G)$ of $G$ is the maximum $h$ such that $G$ contains $h$ vertices of degree at least $h$.  It is not hard to see that 
\[
 \frac{\delta}{2} < \frac{m}{n-1} \leq \alpha(G) \leq h(G) \leq \sqrt{2m}
\]
for every graph $G$, where $\delta$ is the minimum among the degrees of the vertices of $G$.  

For each vertex $v$ of a graph $G$, define $N(v, i) = \{w \in N(v) \mid d(w) = i\}$, i.e., $N(v,i)$ is the set of neighbors of $v$ with degree $i$.  Denote by $L(v)$ the set of neighbors of $v$ of degree at most $d(v)-1$, and $H(v)$ the set of neighbors of $v$ of degree at least $d(v)$, i.e., $L(v) = N(v,1) \cup \ldots \cup N(v, d(v)-1)$, and $H(v) = N(v, d(v)), \ldots, N(v, n-1)$.  We use $\ell(v)$ and $h(v)$ to respectively denote $|L(v)|$ and $|H(v)|$.  Observe that $v$ can have at most $h(G)$ vertices of degree at least $d(v)+1$, thus $h(v) \leq h(G)$, and the number of nonempty sets in the family $N(v,1), \ldots, N(v,d(v))$ is at most $2h(G)$.  

\section{Revisiting the Approach by Chiba and Nishizeki}
\label{sec:chiba-nishizeki}

In~\cite{ChibaNishizekiSJC1985}, Chiba and Nishizeki devised a new method for listing all the triangles of a graph, based on the following lemma.
\begin{lemma}[\cite{ChibaNishizekiSJC1985}]\label{lem:chiba}
 For every graph $G$, \[\sum_{vw \in E(G)}\min\{d(v), d(w)\} \leq 2\alpha(G)m.\]
\end{lemma}
The idea of the algorithm is simple; for each $v \in V(G)$, find those $z \in N(w)$ that are adjacent to $v$, for every $w \in L(v) \cup N(v, d(v))$.  (The original algorithm by Chiba and Nishizeki is slightly different, so as to list each triangle once.  In particular, it requires the vertices of $V(G)$ to be ordered by degree.)  By Lemma~\ref{lem:chiba}, such algorithm takes $O(n+\alpha(G)m)$ time. 

At each iteration, the algorithm by Chiba and Nishizeki finds only those triangles $v,w,z$ such that $d(v) \geq \min\{d(w), d(z)\}$.  Suppose, the aim is to find all those triangles containing some vertex $v$. For instance, suppose we need to dynamically maintain all the triangles of $G$, while vertices are inserted into $G$, and $v$ has been recently inserted. In this case we can find those $z \in H(w)$ that are adjacent to $v$, for every $w \in N(v)$.  The total time required for computing all the triangles with this algorithm while vertices are dynamically inserted is again $O(n+\alpha(G)m)$, according to the next lemmas (see Section~\ref{sec:data structure} for the implementation details).
\begin{lemma}\label{lem:edge traversal}
 Let $e_1, \ldots, e_m$ be an ordering of $E(G)$ for a graph $G$, and call $e_i = v_iw_i$.  Denote by $h_i(v)$ the value of $h(v)$ in the subgraph of $G$ that contains the edges $e_1, \ldots, e_i$, for every $1 \leq i \leq m$.  Then, 
 \[\displaystyle \sum_{i=1}^{m}h_i(v_i) \leq 4\alpha(G)m.\]
\end{lemma}
\begin{proof}
Call $d_i(v)$ to the degree of $v$ in the subgraph of $G$ that contains only the edges $e_1, \ldots, e_i$.  We define two values, $f_i(vw)$ and $F(vw)$, for every pair of vertices $v$ and $w$ that are useful for decomposing the values of $h_i$.  Specifically, for $1 \leq i \leq m$, let
 \begin{equation*}
  f_i(vw) = 
   \begin{cases} 1 & \text{if $v = v_i$, $w \in N_G(v) \setminus \{w_i\}$ and $d_i(w) \geq d_i(v_i)$} \\
   0 &\text{otherwise}
   \end{cases}
 \end{equation*}
 and $F(vw) = \sum_{i=1}^m{f_i(vw)}$.  With this definition,
 \begin{align*}
  \sum_{i = 1}^{m}h_i(v_i) \leq \sum_{i = 1}^{m}\sum_{vw \in E(G)}f_i(vw) = \sum_{vw \in E(G)}{F(vw)}. 
 \end{align*}

 We now prove that $F(vw) \leq 2\min\{d_G(v), d_G(w)\}$.  For this, suppose that $d_G(w) - d_G(v) = k$, for some $k \geq 0$.  By definition, if $e_i \neq vw$ is inserted before $vw$, then $f_i(vw) = 0$.  Similarly, if $e_i \neq vw$ is one of the last $k-1$ edges inserted among those incident to $w$, then $d_i(w) > d_G(v) \geq d_i(v)$, thus $f_i(vw) = 0$.  Therefore, at most $d_G(v)$ edges $e_i$ incident to $w$ are such that $f_i(vw) = 1$, which implies that $F(vw) \leq 2d_G(v)$ as desired.  Consequently, by Lemma~\ref{lem:chiba},
 \begin{align*}
  \sum_{i = 1}^{m}h_i(v_i) \leq \sum_{vw \in E(G)}{F(vw)} \leq \sum_{vw \in E(G)}{2\min\{d_G(v), d_G(w)\}} \leq 4\alpha(G)m.
 \end{align*}
\end{proof}
\begin{lemma}\label{lem:edge traversal static}
  For every graph $G$, $\displaystyle \sum_{vw \in E(G)}h(v) \leq 4\alpha(G)m.$
\end{lemma}
\begin{proof}
 The proof is similar to the one of Lemma~\ref{lem:edge traversal}.  Just replace $d_i$ with $d$ in the definition of $f_i$ and follow the same proof.
\end{proof}

One of the properties about this new algorithm for computing the triangles, is that we can apply it to each vertex $v$ so as to find all the triangles that contain $v$ in $O(d(v)h(G))$ time.  In Section~\ref{sec:4 subgraph counting} we also show how to count the number of some induced graphs on four vertices that contain $v$.

The next corollaries are also relevant for the $h$-graph data structure to be presented.
\begin{corollary}\label{cor:vertex traversal}
 Let $v_1, \ldots, v_n$ be an ordering of $V(G)$, for a graph $G$.  Denote by $h_i(v)$ the value of $h(v)$ in the subgraph of $G$ induced by $v_1, \ldots, v_i$, for every $1 \leq i \leq n$.  Then, 
  \[\displaystyle \sum_{i=1}^{n}\sum_{w \in N(v_i)}h_i(w) \leq 8\alpha(G)m.\]
\end{corollary}
\begin{corollary}\label{cor:vertex traversal static}
  For every graph $G$, $\displaystyle \sum_{v \in V(G)}\sum_{w \in N(v)}h(w) \leq 8\alpha(G)m.$
\end{corollary}

\section{The $h$-Graph Data Structure}
\label{sec:data structure}

In this section we present a new data structure, called the \emph{$h$-graph} data structure, that is well suited for graphs with low $h$-index or low arboricity.  The goal is to make it easier to traverse the set $H(v)$ in $O(h(v))$ time, for every vertex $v$, while vertices and edges are dynamically inserted to and removed from the graph.  Then, this data structure can be used to solve many problems by using the technique implied by Lemmas \ref{lem:edge traversal}~and~\ref{lem:edge traversal static} and Corollaries \ref{cor:vertex traversal}~and~\ref{cor:vertex traversal static}, as the maintenance of triangles described in the previous section.

For each vertex $v$, the $h$-graph data structure stores an object with the following data:
\begin{itemize}
 \item The degree $\mathtt{d}(v)$ of $v$,
 \item A doubly linked list $\mathcal{N}(v)$ containing one object $\mathtt{N}(v,i)$ for each of the nonempty sets of the family $\{N(v, 1), \ldots, N(v, d(v)-1)\}$ ($1 \leq i < d(v)$).  The members of the list $\mathcal{N}(v)$ are ordered so that $\mathtt{N}(v,i)$ appears before $\mathtt{N}(v, j)$ for $1 \leq i < j \leq d(v)-1$.
 \item A set $\mathtt{H}(v)$ representing $H(v)$.
 \item A pointer $\mathtt{data}(v)$, referencing an object that contains the data of $v$.
\end{itemize}
For simplicity, we refer to $\mathtt{H}(v)$ as $\mathtt{N}(v, d(v) + i)$, and to $H(v)$ as $N(v, d(v)+i)$, for every $i \geq 0$.  That is, for every $1 \leq i \leq n-1$, the set $N(v, i)$ contains the neighbors of $v$ with degree $i$, while $\mathtt{N}(v, i)$ is the object of the data structure that contains the neighbors of $v$ with degree $i$.  In the data structure, each $\mathtt{N}(v, i)$ ($1 \leq i \leq n-1$) is stored as a doubly linked list that contains one object for each $w \in N(v,i)$, with the following data associated to the object representing $w$:
\begin{itemize}
 \item a pointer $\mathtt{pos}(v, w)$,  referencing the object of $\mathtt{N}(w, d(v))$ that represents $v$,
 \item a pointer $\mathtt{list}(v, w)$, referencing the list $\mathtt{N}(w, d(v))$,
 \item a pointer $\mathtt{node}(v, w)$, referencing the object that represents $w$, and
 \item a pointer $\mathtt{data}(v, w)$, referencing an object that contains the data of $vw$.  Note: both $\mathtt{data}(v, w)$ and $\mathtt{data}(w, v)$ reference the same object.
\end{itemize}

Observe that the space required for the $h$-graph is $O(n+m)$ bits.  The $h$-graph data structure provides the set of basic operations described below.  Inside the parenthesis we show one value $(x)$ or two values $(x, y)$. The value $x$ is the time required when the operation is applied once, while $y$ is the time required when the operation is applied once for each vertex or edge. (We assume $m > n$.)  For instance, the insertion of one vertex takes $O(dh)$ time, while the iterative insertion of the $n$ vertices takes $O(\alpha m)$.  Here $d = d(v)$, $h = h(G)$ and $\alpha = \alpha(G)$.
\begin{itemize}
 \item $\mathtt{vertex\_insert}$: inserts a new vertex $v$ into $G$ with a specified neighborhood $N(v)$ ($O(dh)$, $O(\alpha m)$).
 \item $\mathtt{vertex\_remove}$: removes a vertex $v$ from $G$ ($O(dh)$, $O(\alpha m)$).
 \item $\mathtt{edge\_insert}$: inserts a new edge $vw$ into $G$ ($O(h)$, $O(\alpha m)$).
 \item $\mathtt{edge\_remove}$: removes an edge $vw$ from $G$ ($O(h)$, $O(\alpha m)$).
 \item $\mathtt{adjacent}$: queries if two vertices $v$ and $w$ are adjacent ($O(h)$).
 \item $\mathtt{H}$: returns the set $H(v)$ for a vertex $v$ ($O(1)$).
 \item $\mathtt{N'}$: returns the set $N'(v)$ for a vertex $v$ ($O(dh)$, $O(\alpha m)$)
 \item $\mathtt{G[N()]}$: returns the adjacency lists of the graph $G[N(v)]$ for a vertex $v$ ($O(dh)$, $O(\alpha m)$).
\end{itemize}
The implementation of $\mathtt{H}$ is trivial.  We discuss the other operations below.  In the following sections we show several algorithms that work on the $h$-graph data structure.

\paragraph{The insertion of vertices and edges.} The algorithm for inserting a new edge $vw$ into $G$ is straightforward.  In a first phase, update the families $\mathcal{N}(z)$ for every $z \in N[v] \cup N[w]$.  For this, create the set $\mathtt{N}(v, d_G(v))$, move the vertices with degree $d_G(v)$ from $\mathtt{H}(v)$ to $\mathtt{N}(v, d_G(v))$, and move $v$ from $\mathtt{N}(z, d_G(v))$ to $\mathtt{N}(z, d_G(v)+1)$, for every $z \in H(v)$.  Next, apply the analogous operations for $w$.  The second phase is to actually insert the edge $vw$.  For this, insert $v$ at the end of $\mathtt{N}(w, d_G(v)+1)$ and $w$ at the end of $\mathtt{N}(v, d_G(w)+1)$, update the values of $\mathtt{d}(v)$ and $\mathtt{d}(w)$, and create the pointers $\mathtt{pos}$, $\mathtt{list}$, $\mathtt{node}$, and $\mathtt{data}$ for $vw$.  

Discuss the time complexity of the above algorithm.  For the first phase, apply Algorithm~\ref{alg:edge insert} twice, once for $v$ and once for $w$.  Recall that this algorithm is applied before incrementing $\mathtt{d}$ for $v$ and $w$, thus $\mathtt{d}(z)$ is the degree of $z$ before the insertion of $vw$.  Note that each iteration of the main loop can be implemented so as to run in $O(1)$ time, by using the pointers in the data structure.  Thus, the update of $\mathcal{N}(z)$, for every $z \in N[v] \cup N[w]$, takes $O(h(v) + h(w))$ time.  
\begin{algorithm}\label{alg:edge insert}
 Update of $\mathcal{N}(z)$ for every $z \in N[v]$.
 \begin{enumerate}
  \item Insert a new empty set $\mathtt{N}(v,d(v))$ at the end of $\mathcal{N}(v)$.
  \item For each $z \in \mathtt{H}(v)$:
  \item\hspace*{5mm}If $\mathtt{d}(z) = \mathtt{d}(v)$, then move $z$ from $\mathtt{H}(v)$ to $\mathtt{N}(v, d(v))$.
  \item\hspace*{5mm}Move $v$ from $\mathtt{N}(z, d(v))$ to $\mathtt{N}(z, d(v)+1)$.  If $\mathtt{N}(z,d(v)) = \emptyset$, then delete $\mathtt{N}(z,d(v))$.
  \item If $\mathtt{N}(v, d(v)) = \emptyset$, then delete $\mathtt{N}(v, d(v))$.
 \end{enumerate}
\end{algorithm}
For the second phase, traverse the family $\mathcal{N}(w)$ until the first set $\mathtt{N}(w, d)$ with $d > d_G(v)$ is reached.  Then, create the set $\mathtt{N} = \mathtt{N}(w,d_G(v)+1)$, if $d > d(v)+1$, and insert $v$ into $\mathtt{N}$.  Next, traverse $\mathcal{N}(v)$ so as to find the set that must contain $w$, and insert $w$. Recall that there are at most $2h(G)$ sets inside each of $\mathcal{N}(v)$ and $\mathcal{N}(w)$.  So, the time required by these steps is $O(\min\{d(v), d(w), h(G)\})$, while the creation of the pointers and the increase of $\mathtt{d}$ take $O(1)$ time. Therefore, the insertion of $vw$ requires time
\[
 O\left(\min\{d(v), d(w), h(G)\} + h(v) + h(w)\right) = O(h(G))
\]

The insertion of a new vertex $v$ is simple.  First, insert $v$ as an isolated vertex, and then add the edges $vw$, for each $w \in N(v)$.  The time required is $O(1 + d(v)h(G))$.  If we use the above algorithm for building $G$ from scratch, then the total time is $O(n + \alpha(G)m)$, by Lemma~\ref{lem:edge traversal}.

\begin{corollary}
 The time required for inserting the vertices and edges of a graph $G$, one at a time in no particular order, into an initially empty $h$-graph data structure is $O(n + \alpha(G) m)$.
\end{corollary}

\begin{proof}
 The insertion of the $n$ vertices takes $O(n)$ time.  Let $e_1, \ldots, e_m$ be an ordering of $E(G)$, and call $d_i(v)$ and $h_i(v)$ to the values of $d(v)$ and $h(v)$ in the graph prior the insertion of $e_i$.  By the analysis above, the time required for the insertion of $e_i = vw$ is \[O(\min\{d_i(v), d_i(w)\} + h_i(v) + h_i(w)).\]  Thus, by Lemmas \ref{lem:chiba}~and~\ref{lem:edge traversal}, the total time required for the insertion of all the edges is
 \begin{align*}
  \sum_{vw \in E(G)}{O(\min\{d_i(v), d_i(w)\}) + \sum_{vw \in E(G)}O(h_i(v) + h_i(w)) = O(\alpha(G)m)}
 \end{align*}
\end{proof}

\paragraph{The removal of vertices and edges.}  For the removal of an edge $vw$, we should undo the insertion process.  But this time, we first undo the second phase, and then we undo the first phase.  To undo the second phase we need to physically remove the edge $vw$.  Suppose that $d_G(v) \leq d_G(w)$, i.e., $w \in H(v)$.  Traverse $\mathtt{H}(v)$ so as to locate and remove the object that represents $w$ in $\mathtt{H}(v)$.  By using the $\mathtt{pos}$ and $\mathtt{list}$ pointers, we can easily remove $v$ from $\mathtt{N}(w, d_G(v))$ in $O(1)$ time.  Since $v$ has $h(v)$ neighbors in $H(v)$, this phase takes $O(h(v)) = O(h(G))$ time.  To undo the first phase, we need to update the families $\mathcal{N}(z)$ for every $z \in N_G(v) \cup N_G(w)$.  For this, we move $v$ from $N(z, d_G(v))$ to $N(z, d_G(v)-1)$ for every $z \in H_G(v)$, and then we remove $\mathtt{N}(v, d_G(v)-1)$ from $\mathcal{N}(v)$ so as to append it to $\mathtt{H}(v)$.  Note that in this step we should update all the \texttt{list} pointers of the vertices in $\mathtt{N}(v, d_G(v) -1)$ so as to reference $\mathtt{H}(v)$.  Next, we should apply the same operations for $w$.  This phase is rather similar to the one for edge insertion and it also takes time \[
 O(h_{G \setminus \{vw\}}(v) + h_{G \setminus \{vw\}}(w)) = O(h(G)).
\]  

As before, we can remove a vertex $v$ by removing all its incident edges first, and the object representing $v$ later. The time required by this operation is again $O(1 + d(v)h(G))$.  Finally, if we use the above algorithm for decomposing $G$, then the time required for removing all the edges is $O(n + \alpha(G)m)$, by Lemma~\ref{lem:edge traversal}.

\begin{corollary}
 The time required for removing the vertices and edges of a graph $G$, one at a time and in no particular order, from an $h$-graph data structure is $O(n + \alpha(G) m)$.
\end{corollary}

\paragraph{Adjacency query.}  Querying whether two vertices $v$ and $w$ are adjacent takes $O(h(v) + h(w)) = O(h(G))$ time.  Simply traverse the set $H(z)$ for $z \in \{v,w\}$ with minimum degree.

\paragraph{Edge-neighborhoods and their induced subgraphs.}  To compute the set $N'(v)$ for a vertex $v$, first mark each $z \in N(v)$ with $1$.  Following, traverse each $z \in H(w)$ for every $w \in N(v)$ and, if $z$ is marked with $1$, then insert it into $N'(v)$ and mark it with $2$.  Since each $w$ is traversed in $O(h(w)) = O(h(G))$ time, this algorithm takes $O(d(v)h(G))$ time.  Furthermore, the time required to find $\{N'(v)\}_{v \in V(G)}$, by applying this algorithm to all the vertices in the graph, is $O(\alpha(G) m)$, by Corollary~\ref{cor:vertex traversal static}.

Clearly, the subgraph of $G$ induced by $N(v)$ is just the graph whose vertex set is $N(v)$ and whose edge set is $N'(v)$.  Thus, the graph $G[N(v)]$, implemented with adjacency lists, can be computed in $O(d(v)h(G))$ time, while the family $\{G[N(v)]\}_{v\in V(G)}$ can be computed in $O(\alpha(G) m)$ time.  

\section{Applications}
\label{sec:applications}

In this section we list several problems that can be improved by using the technique implicit in Lemma~\ref{lem:edge traversal} and Corollary~\ref{cor:vertex traversal}.  These algorithms build upon the $h$-graph data structure, and serve as examples of its applicability.  One of the most appealing aspects of some of these algorithms is that they are simple to obtain from the problems' definitions.

\subsection{Listing the Cliques of a Vertex}

In~\cite{ChibaNishizekiSJC1985}, Chiba and Nishizeki devised an algorithm for listing all the $K_k$'s of a graph, for a given $k \in \mathbb{N}$.  Conceptually, their algorithm computes all the $K_{k-1}$'s in the subgraph of $G$ induced by $N(v)$, and it outputs $v$ plus these cliques, for every $v \in G$. (Again, some ordering of the vertices is required so as to avoid repetitions.)

This algorithm can be translated to find all the $K_k$'s that contain a given vertex $v$.  First, compute $G' = G[N(v)]$ in $O(d(v)h(G))$ time, as in Section~\ref{sec:data structure}. Then, apply the algorithm by Chiba and Nishizeki that lists the $K_{k-1}$'s of $G'$ in $O(|V(G')|+(k-1)\alpha(G')^{k-1}|E(G')|)$ time.  Note that
\begin{align*}
|E(G')| \leq&\ \sum_{w \in N(v)}\min\{d(v), d(w)\} \leq \\
       \leq&\ \sum_{\substack{w \in N(v) \\ d(w) \leq h(G)}}d(w) + \sum_{\substack{w \in N(v) \\ d(w) > h(G)}}d(v) = O(d(v)h(G)),
\end{align*}
thus the time required to find all the $K_k$'s that contain $v$ is
\[
O\left(d(v)h(G) + |V(G')| + k\alpha(G')^{k-1}|E(G')|\right) = O\left(kd(v)h(G)\alpha(G)^{k-1}\right).
\]

Also, the total time required for listing all the cliques, by iteratively executing the above algorithm for each vertex $v$ (and then removing $v$), is 
\begin{align*}
O\left(n+\alpha(G)m + k\alpha(G)^{k-1}\sum_{v \in V(G)}\sum_{w \in N(v)}\min\{d(v), d(w)\}\right) = O\left(n+k\alpha(G)^km\right)
\end{align*}
matching the time complexity of the algorithm by Chiba and Nishizeki.

\subsection{The $4$-Subgraph Counting Problem}
\label{sec:4 subgraph counting}

The $4$-subgraph counting problem is the problem of counting how many copies of some graph $H$ on four vertices appear as induced subgraphs of a given graph $G$.  The connected graphs on four vertices are six: the complete graph $K_4$, the diamond $K_4 \setminus \{e\}$ for $e \in E(K_4)$, the square $C_4$, the path $P_4$, the paw $\overline{P_3 \cup K_1}$, and the claw $\overline{K_3 \cup K_1}$.  The disconnected graphs on four vertices are five: $\overline{K_4}$, $\overline{K_4 \setminus \{e\}}$, $\overline{C_4}$, $P_3 \cup K_1$, and $K_3 \cup K_1$. In~\cite{KloksKratschMullerIPL2000}, Kloks et al.\ showed a system of linear equations that solves the $4$-subgraph counting problem for connected graphs.  Specifically, Kloks et al.\ proved the following theorem.
\begin{theorem}[\cite{KloksKratschMullerIPL2000}]\label{thm:4-subgraph Kloks}
 Let $\tilde{H}$ be a connected graph on four vertices such that there is an $O(t(G))$ time algorithm counting the number of induced $\tilde{H}$'s in a graph $G$. Then, there is an $O(n^\omega + t(G))$ time algorithm counting the number of induced $H$'s of $G$ for all connected graphs $H$ on four vertices.
\end{theorem}
Since the number of $K_4$'s can be computed in either $O(n+m^{(\omega+1)/2}) = O(n+m^{1.61})$ or $O(n+\alpha(G)^2m)$ time~\cite{ChibaNishizekiSJC1985,KloksKratschMullerIPL2000}, solving the $4$-subgraph counting problem for connected graphs takes $O(n^\omega + \min\{m^{1.61},\alpha(G)^2m\})$ time.  This result can be improved with a new system of linear equations that solves the $4$-subgraph counting problem in $O(n+\alpha(G)m + t(G))$ time, even for disconnected graphs.  By using the algorithm in~\cite{ChibaNishizekiSJC1985} for counting the number of $K_4$'s, an $O(n+\min\{m^{1.61},\alpha(G)^2m\})$ time algorithm is obtained.  The system of linear equations appears in the proof of the next theorem.
\begin{theorem}\label{thm:4-subgraphs}
 Let $\tilde{H}$ be a graph on four vertices such that there is an $O(t(G))$ time algorithm counting the number of induced $\tilde{H}$'s in a graph $G$. Then, there is an $O(n+\alpha(G)m + t(G))$ time algorithm counting the number of induced $H$'s of $G$ for every graph $H$ on four vertices.
\end{theorem}
\begin{proof}
 Let $k$, $d$, $s$, $p$, $q$, and $y$ denote the number of induced $K_4$'s, diamonds, squares, $P_4$'s, paws, and claws in $G$, respectively  Similarly, let $\bar{k}$, $\bar{d}$, $\bar{s}$, $\bar{q}$ and $\bar{y}$ be the number of induced complements of $K_4$'s, diamonds, squares, paws, and claws in $G$, respectively.  Define $\overline{m} = \binom{n}{2} - m$, $\overline{d(v)} = n-d(v)-1$ for $v \in V(G)$, and $\delta(v,w) = d(v) - d(vw)$ for $vw \in E(G)$.  That is, $\overline{m}$ is the number of edges of $\overline{G}$, $\overline{d(v)}$ is the degree of $v$ in $\overline{G}$, and $\delta(v,w)$ is the number of vertices that are adjacent to $v$ and not $w$.  Finally, let $\mathcal{S}$ be the set obtained after executing the algorithm $C4$ in~\cite{ChibaNishizekiSJC1985}.  In $\mathcal{S}$, each element is a triple $(v, w, L)$ where $v,w \in V(G)$ and $L$ is a set of vertices.  Then, $G$ fulfills the following system of linear equations.
\begin{align}
 & \sum_{(v,w,L) \in \mathcal{S}}{\binom{|L|}{2}} = 3k + d + s \label{eq:1}  \\
 & \sum_{vw \in E(G)}{\binom{d(vw)}{2}} = 6k+d \label{eq:2} \\ 
 & \sum_{vw \in E(G)}{\delta(v,w)\delta(w,v)} = 4s + p \label{eq:3} \\
 & \sum_{vw \in E(G)}\left({\binom{\delta(v,w)}{2} + \binom{\delta(w,v)}{2}}\right) = q + 3y \label{eq:4} \\
 & \sum_{vw \in E(G)}{\binom{d(v) + d(w) - d(vw) - 2}{2}} = 6k + 5d + 4s + p + 3q + 3y \label{eq:5} \\
 & \sum_{v \in V(G)}{d'(v)(n-3)} = 12k + 6d + 3q + 3\bar{y} \label{eq:6} \\
 & \sum_{v \in V(G)}{\binom{d(v)}{2} (n-3)} = 12k + 8d + 4s + 2p + 5q + 3y + \bar{q} + 3\bar{y} \label{eq:7} \\
 & \binom{m}{2} - \sum_{v \in V(G)}{\binom{d(v)}{2}} = 3k + 2d + 2s + p + q + \bar{s} \label{eq:8} \\
 & \binom{\overline{m}}{2} - \sum_{v \in V(G)}{\binom{\overline{d(v)}}{2}} = s + p + 3\bar{k} + 2\bar{d} + 2\bar{s} + \bar{q} \label{eq:9} \\
 & \binom{n}{4} = k + d + s + p + q + y + \bar{k} + \bar{d} + \bar{s} + \bar{q} + \bar{y}\label{eq:10}
\end{align}
Equations~(\ref{eq:2})--(\ref{eq:4}) correspond to the first, third and fifth equations of~\cite{KloksKratschMullerIPL2000}, respectively.  For the translation between them, observe that if $A$ is the adjacency matrix of $G$ and $C$ is the adjacency matrix of the complement of $G$, then $A^2_{v,w} = d(vw)$ and $AC_{v,w} = \delta(v,w)$.  Thus, $G$ satisfies these equations~\cite{KloksKratschMullerIPL2000}.  

Each triple of $\mathcal{S}$ represents a set of non induced cycles of $G$ in the following way.  Let $v_1, \ldots, v_n$ be an ordering of $V(G)$ in non increasing order of degree.  Then $(v_i, v_j, L)$ is a triple of $\mathcal{S}$ if and only if $i < \min\{j, k, l\}$, $L \subset N(v_i) \cap N(v_j)$ and $|L| > 2$.  Therefore, $\binom{|L|}{2}$ counts the number of non induced $4$-length cycles that contain $v_i$ and $v_j$ and such that
$v_iv_j$ is not an edge of the cycle.  Such cycles count three times each $K_4$, and once each diamond and each square.  Thus, (\ref{eq:1}) is fulfilled by $G$.

The remaining equations follow analogously, by observing that: (\ref{eq:5}) counts, for each edge $vw$, the number of pairs of vertices that are adjacent to at least one of $v$ and $w$; (\ref{eq:6}) counts, for each vertex $v$, the number of triangles of $v$ plus one vertex; (\ref{eq:7}) counts those graphs on four vertices where $v$ has degree at least $2$; (\ref{eq:8})~and~(\ref{eq:9}) count the number of pair of disjoint edges on $G$ and $\overline{G}$, respectively; and~(\ref{eq:10}) counts the number of induced graphs on $4$ vertices of $G$.

These 10 equations are linearly independent, and by fixing the number of subgraphs $\tilde{H}$ we can solve the whole system.  As for the time complexity of computing the constant terms of the system, algorithm $C4$ takes $O(n+\alpha(G)m)$ time~\cite{ChibaNishizekiSJC1985}; $d' = |N'|$ is found in $O(n+\alpha(G)m)$ time; the degrees of the edges and $\delta$ are found in $O(n+\alpha(G)m)$ by traversing all edge-neighborhoods; and all the terms in (\ref{eq:7})--(\ref{eq:10}) are computed in $O(n+m)$ time.  
\end{proof}

We now consider a slight modification of the $4$-subgraph counting problem.  Given a vertex $v$, the goal is to count the number of graphs on $4$ vertices that contain $v$.  We focus our attention on four types of connected graphs: $K_4$'s, diamonds, paws, and claws.  For $i \in \{1,2,3\}$, define $k_i(v)$, $d_i(v)$, $q_i(v)$, and $y_i(v)$ as the number of $K_4$'s, diamonds, paws, and claws that contain a given vertex $v$, where the degree of $v$ in such an induced subgraph is $i$.  In the previous section we saw that $k_3(v)$ can be computed in $O(d(v)h(G)\alpha(G))$ time.  The following theorem shows that we can compute $d_i(v)$, $q_i(v)$, and $y_i(v)$, once $k_3(v)$ is given.
\begin{theorem}\label{thm:4-subgraph v}
 There is dynamic graph data structure, consuming $O(n+m)$ space, with the following properties:
 \begin{itemize}
  \item Both the insertion and the removal of a vertex $v$ take $O(d(v)h(G))$ time.
  \item The time required for inserting the vertices of $G$, one at a time, into an initially empty instance of the data structure is $O(n + \alpha(G) m)$.
  \item Given $k_3(v)$, the values of $d_i(v)$, $q_i(v)$, and $y_i(v)$ can be found in $O(d(v)h(G))$ time, for every $v \in V(G)$ and $i \in \{1,2,3\}$.
  \item If $k_3(v)$ is given for every $v \in V(G)$, then the number of diamonds, paws, and claws can be found in $O(n+\alpha(G)m)$ time.
 \end{itemize}
\end{theorem}
\begin{proof}
 Fix a vertex $v$ and let $k_i = k_i(v)$, $d_i = d_i(v)$, $q_i = q_i(v)$, and $y_i = y_i(v)$, for $i \in \{1,2,3\}$.  Define $\delta$ as in Theorem~\ref{thm:4-subgraphs}.  Then, $v$ fulfills the following system of linear equations.
\[
\begin{array}{rL@{\hspace{1cm}}rL}
 d_3 &= \sum_{w \in N(v)}{\binom{d(vw)}{2}} - 3k_3 & d_2 &= \sum_{wz \in N'(v)}{(d(wz)-1)} - 3k_3 \\
 2q_3 &= \sum_{w \in N(v)}{d(vw)\delta(v,w)} - 2d_3 & q_2 &= \sum_{w \in N(v)}{d(vw)\delta(w,v)} - 2d_2 \\
 q_1 & \multicolumn{3}{L}{= \sum_{w \in N(v)}(d'(w) - d(vw)) - 3k_3 - 2d_2}  \\ 
 3y_3 &= \sum_{w \in N(v)}{\binom{\delta(v,w)}{2}} - q_3 & y_1 &= \sum_{w \in N(v)}{\binom{\delta(w,v)}{2}} - q_1
\end{array}
\]
These equations are similar to those in Theorem~\ref{thm:4-subgraphs}, e.g., $\binom{d(vw)}{2}$ counts the number of pairs of vertices $x,z$ that are both adjacent to $v$ and $w$, thus $\sum_{w \in N(v)}{\binom{d(vw)}{2}}$ counts three times each $K_4$ that contains $v$, and one time each diamond that contains $v$ with degree $3$.  The other equations follow analogously.

The dynamic data structure is just the $h$-graph data structure with the addition that $d(wz)$ is stored in $\mathtt{data}(w, z)$, for every edge $wz$.  This value can be easily updated in $O(d(v)h(G))$ time when a vertex $v$ is either inserted or removed.  Indeed, according to whether $v$ is inserted or removed, we can increase or decrease in $1$ the value of $d(wz)$ for every $wz \in N'(v)$ with a single traversal of $N'(v)$.  On the other hand, the value of $d(vw)$ is simply the degree of $w$ in $G[N(v)]$, for every $w \in N(v)$.  Then, the values of $d_i(v)$, $q_i(v)$, and $y_i(v)$ can be obtained in $O(d(v)h(G))$ time by solving the system of equations above, once $k_3(v)$ is given.  Furthermore, by solving the above equations for all the vertices, we can compute the number of diamonds, paws, and claws can be found in $O(n+\alpha(G)m)$ time, by Corollary~\ref{cor:vertex traversal static}.
\end{proof}

\subsection{Dynamic Recognition of Diamond-Free Graphs}

Recall that the graph that is obtained by removing one edge from a complete graph of four vertices is called a \emph{diamond}.  A \emph{diamond-free} graph is a graph that contains no induced diamond.  Diamond-free graphs appear in many contexts; for example in the study of perfect graphs~\cite{Conforti1989,FonluptZemirlineRMM1993,TuckerJCTSB1987}.  

In~\cite{KloksKratschMullerIPL2000}, Kloks et al.\ showed how to find an induced diamond in $O(m^{3/2} + n^\omega)$ time, if one exists.  The fast matrix multiplication algorithm is used in one of the steps of this algorithm, which explains why $n^\omega$ is a term of the complexity order.  However, the fast matrix multiplication can be avoided while improving the time complexity to $O(m^{3/2})$ time, as it is shown by Eisenbrand and Grandoni~\cite{EisenbrandGrandoniTCS2004}.  Talmaciu and Nechita~\cite{TalmaciuNechitaI2007} devised a recognition algorithm based on decompositions, but they claim that in the worst case the time required by their algorithm is not better than the one by Kloks et al. Note that Theorem~\ref{thm:4-subgraphs} implies that there is an $O(\alpha(G)^2m)$ time algorithm for recognizing whether a graph is a diamond-free graph, improving over the previous algorithms for some sparse graphs.  Finally, Vassilevska~\cite{Vassilevska2008} used the algorithm by Eisenbrand and Grandoni to find an induced $K_k \setminus e$ in a graph.  A $K_k \setminus e$ is a complete graph on $k$ vertices, minus one edge.  For every even $k$, the algorithm by Vassilevska takes $O\left(d(n,m)m^{(k-4)/2}\right)$ time, where $d(n,m)$ is the time required to find a diamond in a graph with $n$ vertices and $m$ edges.  Thus, this algorithm is implicitly improved with each improvement on $d(n,m)$.  

The algorithm by Kloks et al.\ is based on the fact that a graph $G$ is diamond-free if and only if $G[N(v)]$ is a disjoint union of maximal cliques, for every $v \in V(G)$.  Testing whether a graph is a disjoint union of cliques takes linear time, and we saw in Section~\ref{sec:data structure} how to compute the family $\{G[N(v)]\}_{v \in V(G)}$ in $O(\alpha(G) m)$ time.  Therefore, by using the $h$-graph data structure, the algorithm by Kloks et al.\ can be implemented so as to run in $O(\alpha(G) m)$ time, improving over the algorithm by Eisenbrand and Grandoni and the algorithm implied by Theorem~\ref{thm:4-subgraphs}.  

In this section we develop a dynamic data structure for maintaining diamond-free graphs.  This data structure can also be used to find an induced diamond of a static graph $G$ in $O(\alpha(G)m)$ time, if one exists.  As a by-product, the data structure can be used to query the maximal cliques of the dynamic diamond-free graph in constant time.  The data structure is based on this well known theorem about diamond-free graphs.

\begin{theorem}\label{thm:uniclical}
 A graph is a diamond-free graph if and only if every edge belongs to exactly one maximal clique.
\end{theorem}

An $h$-graph data structure is used to represent a diamond-free graph $G$.  Also, the family of non singleton cliques of $G$ is stored in the dynamic diamond-free data structure.  Each non singleton clique is in turn implemented as the set of edges that belong to the clique.  Also, the pointer $\mathtt{data}(vw)$ references the unique clique $C_{vw}$ of the family that contains $vw$, for every $vw \in E(G)$.  Finally, each clique $C$ is associated with a counter $c(C)$ that is initialized to $0$.  The purpose of $c(C)$ is to count the number of neighbors of $v$ inside $C$, when a vertex $v$ is inserted into $G$.

We are now ready to discuss the operations allowed by the data structure.  Suppose first that $G$ is a diamond-free graph and $v \not \in V(G)$ is to be inserted into $G$.  Vertex $v$ is given with its neighborhood set $N(v)$, which also defines its edge-neighborhood $N'(v)$.  Say that $v$ is \emph{edge-adjacent} to a clique $C$ of $G$ when $E(C) \cap N'(v) \neq \emptyset$, while $v$ is \emph{fully edge-adjacent} to $C$ when $E(C) \cap N'(v) = E(C)$.  The following theorem shows how to insert $v$ into $G$.

\begin{theorem}\label{thm:G cup v diamond-free}
 The graph $G \cup \{v\}$ is diamond-free if and only if the following two statements hold for every maximal clique $C$ to which $v$ is edge-adjacent. 
\begin{enumerate}
 \item $v$ is fully edge-adjacent to $C$, and
 \item if $v$ is edge-adjacent to a maximal clique $C' \neq C$, then $V(C') \cap V(C) = \emptyset$.
\end{enumerate}
\end{theorem}

\begin{proof}
 If $v$ is not fully edge-adjacent to $C$, then there is some vertex $u \in C$ which is not adjacent to $v$.  Since $v$ is edge-adjacent to $C$, then there is an edge $wz \in N'(v) \cap E(C)$.  But then, $u, v, w, z$ induce a diamond in $G$.  Suppose now that $v$ is edge-adjacent to a maximal clique $C' \neq C$ that contains some vertex $u \in C$.  Since $C$ and $C'$ are maximal cliques, it follows that there is some vertex $w \in C$ which is not adjacent to a vertex $z \in C'$.  Thus, $u, v, w, z$ induce a diamond.

 For the converse, suppose that $G \cup \{v\}$ is not a diamond-free graph. Since $G$ is diamond-free, then there are three vertices $u, w, z$ such that together with $v$ induce a diamond in $G$.  If $u$ and $v$ are not adjacent, then $u, w, z$ belong to some maximal clique $C$ of $G$.  Since $v$ is adjacent to $w,z$ but not to $u$, we obtain that $v$ is edge-adjacent but not fully edge-adjacent to $C$.  So, we may assume that $v$ is adjacent to $u, w, z$, and that $w$ and $z$ are not adjacent.  But then, $uw$ and $uz$ belong to different maximal cliques $C$ and $C'$.  Thus, $v$ is edge-adjacent to $C$ and $C'$ and $u \in V(C) \cap V(C')$.
\end{proof}

Algorithm~\ref{alg:diamond-free}, which is obtained from Theorem~\ref{thm:G cup v diamond-free}, can be used to decide whether $G \cup \{v\}$ is a diamond-free graph.  

\begin{algorithm}\label{alg:diamond-free}
Insertion of a vertex $v$ into a diamond-free graph $G$
\begin{enumerate}
 \item Remove the mark from every $w \in V(G)$.\label{alg:diamond-free:step1}
 \item For each maximal clique $C$ of $G$ to which $v$ is edge-adjacent:\label{alg:diamond-free:step2}
 \item\hspace*{5mm} If $v$ is not fully edge-adjacent to $C$, then output ``$G \cup \{v\}$ is not diamond-free.''\label{alg:diamond-free:step3}
 \item For each maximal clique $C$ of $G$ to which $v$ is edge-adjacent:\label{alg:diamond-free:step4}
 \item\hspace*{5mm} If there is some marked $w \in C$, then output ``$G \cup \{v\}$ is not diamond-free.''\\\hspace*{5mm} Otherwise, mark every $w \in C$.\label{alg:diamond-free:step5}
 \item Output ``$G \cup \{v\}$ is a diamond-free graph''.\label{alg:diamond-free:step6}
\end{enumerate}
\end{algorithm}

Observe that after $G \cup \{v\}$ is claimed to be a diamond-free graph in Step~\ref{alg:diamond-free:step6}, all the vertices of every edge-adjacent maximal clique $C$ have a mark, and every vertex $w \in C$ was traversed and marked only once.  Then, $v$ is fully-adjacent to $C$ and no vertex of $C$ belongs to other clique to which $v$ is edge-adjacent.  Therefore, by Theorem~\ref{thm:G cup v diamond-free}, the algorithm is correct.

Discuss the implementation of Algorithm~\ref{alg:diamond-free}.  The input of the algorithm is formed by the graph $G$, the vertex $v$, and the set $N(v)$ of neighbors of $v$ in $G$.  The first step is to insert $v$ into the $h$-graph data structure of $G$ in $O(d(v)h(G))$ time.  We consider that $w$ is marked if and only if $\mathtt{data}(vw) \neq NULL$.  Hence, Step~\ref{alg:diamond-free:step1} of the algorithm is executed in $O(d(v))$ time, by traversing once the set $N(v)$ to initialize $\mathtt{data}(vw)$.  

Before traversing each edge-adjacent maximal clique in Steps \ref{alg:diamond-free:step2}~and~\ref{alg:diamond-free:step4}, we compute $N'(v)$, in $O(d(v)h(G))$ time, as in Section~\ref{sec:data structure}.  The family $\mathcal{C} = \{C_{wz}\}_{wz \in N'(v)}$ is computed by iteratively inserting $C_{wz}$ into $\mathcal{C}$, for every $wz \in N'(v)$.  While $\mathcal{C}$ is generated, $|N'(v) \cap E(C_{wz})|$ can be computed by increasing $c(C_{wz})$ by $1$ when $wz$ is first traversed.  This operation takes $O(1)$ time per $wz \in N'(v)$, thus it takes $O(d(v)h(G))$ total time.  Once $\mathcal{C}$ is computed, each maximal clique $C \in \mathcal{C}$ is traversed.  

If $c(C) = |N'(v) \cap E(C)| \neq |E(C)|$, then $v$ is not fully edge-adjacent to $C$ and the algorithm stops with a failure message in Step~\ref{alg:diamond-free:step3}.  Otherwise, every maximal clique $C \in \mathcal{C}$ has to be traversed once again to update $\mathtt{data}(vw)$ so as to reference $C_{vw}$, for every $w \in C$.  Instead of doing this, we traverse $N'(v)$ and, for every edge $wz$, we set $\mathtt{data}(vw)$ and $\mathtt{data}(vz)$ to point to $C_{wz}$, because $v$ is fully adjacent to $C_{wz}$.  If when updating $C_{vw}$ (resp.\ $C_{vz}$) to point to $C$ we discover that $C_{vw}$ (resp.\ $C_{vz}$) points to $C'$ with $C \neq C'$, then the algorithm stops in failure as in Step~\ref{alg:diamond-free:step5}.  If $C_{wz}$ is maintained together with $wz$ in $N'(v)$, then each edge gets traversed in $O(1)$ time.  Thus, the time complexity of both loops is $O(d'(v)) = O(d(v)h(G))$.

After the main loops are completed, the algorithm claims that $G \cup \{v\}$ is a diamond-free graph.  Before another vertex can be inserted, the data structure that represents $G$ has to be updated into a data structure that represents $G \cup \{v\}$.  The set of edges $\{vw \mid w \in N(v)\}$ can be split into two types.  Those whose pointer $\mathtt{data}$ references some maximal clique of $G$ and those whose $\mathtt{data}$ is still $NULL$.  For those edges that reference some maximal clique, the edge has to be inserted into $C_{vw}$.  For those edges $vw$ with $\mathtt{data} = NULL$, a new clique has to be created that contains only the edge $vw$.  This clique has to be inserted to the set of maximal cliques of $G \cup \{v\}$ and $\mathtt{data}(vw)$ has to be updated accordingly.  All this can be done in $O(d(v))$ time.  Finally, $c(C)$ has to be updated to $0$ for every $C \in \mathcal{C}$, again in $O(d'(v)) = O(d(v)h(G))$ time.  Therefore, the insertion of a new vertex takes $O(d(v)h(G))$ time as desired.

Observe that if $G \cup \{v\}$ is a diamond-free graph, then its family of maximal cliques is obtained in $O(1)$ time and the maximal clique to which an edge $vw$ belongs can be queried in $O(1)$ time.   Algorithm~\ref{alg:diamond-free} can also be modified so as to output an induced diamond in $O(d(v)h(G))$ time, when $G \cup \{v\}$ is not diamond-free.  Consider the two alternatives for the algorithm to stop in failure.  First, if $v$ is not fully edge-adjacent to $C \in \mathcal{C}$, then there is some vertex $w \in C$ that is not adjacent to $v$.  In this case, $v, w$ and the endpoints of an edge in $N'(v) \cap C$ induce a diamond.  To find $w$, traverse the edges of $C$ and query if each endpoint is adjacent to $v$.  (For this step, we may use the marks used in the computation of $N'(v)$, or we may re-compute them with a single traversal of $N(v)$.)  The first vertex that is not adjacent to $v$ is taken as $w$.  To find an edge in $N'(v) \cap C$, traverse $N'(v)$ until some edge of $C$ is reached.  Thus, the certificate in this case can be found in $O(d'(v))$ time.  The second reason for the algorithm to stop in failure is that $v$ is fully edge-adjacent to $C$ and $C'$, and $vw$ is marked with $C'$ while trying to mark it with $C \neq C'$.  In this case, we ought to find two non adjacent vertices of $C \cup C'$.  As in the proof of Theorem~\ref{thm:G cup v diamond-free}, these two vertices together with $v$ and $w$ induce a diamond.  To find these two vertices observe that $w$ is the unique vertex in $C \cap C'$ and that all the vertices of $C \setminus \{w\}$ are not adjacent to the vertices of $C' \setminus \{w\}$.  Thus, the two non adjacent vertices of $C \cup C'$ can be found in $O(1)$ time by traversing at most one edge of both $C$ and $C'$.

The dynamic data structure supports also the removal of vertices and the insertion and removal of edges.  For the removal of a vertex $v$, note that $G \setminus \{v\}$ is always a diamond-free graph.  Thus, all we need to do is to remove $v$ from the $h$-graph data structure, and to remove $vw$ from $C_{vw}$, for every $w \in N(v)$.  As explained in Section~\ref{sec:data structure}, the former operation takes $O(d(v)h(G))$ time, while the latter takes $O(d(v))$ time.  With respect to the insertion of an edge $vw$, graph $G \cup \{vw\}$ is diamond-free if and only if $|N(v) \cap N(w)| \leq 1$ and, if there is some $z \in N(v) \cap N(w)$, then $d_G(vz) = d_G(wz) = 0$.  We can compute $|N(v) \cap N(w)|$ in $O(d(v) + d(w))$ time, while, for $z \in N(v) \cap N(w)$, we can access $\mathtt{data}(vz)$ and $\mathtt{data}(wz)$ in $O(1)$ time so as to see whether these maximal cliques have exactly one edge.  If so, the maximal cliques $C_{vz}$ and $C_{wz}$ need to be merged, again in $O(1)$ time.  Thus, the insertion of $vw$ takes $O(d(v) + d(w))$ time.  Finally, for $vw \in E(G)$, the graph $G \setminus \{vw\}$ is diamond-free if and only if $d(vw) \leq 1$, i.e., if $C_{vw}$ has at most three edges.  If $d(vw) = 1$, then $C_{vw}$ has to be split into two maximal cliques of $G \setminus \{vw\}$ in $O(1)$ time.  Therefore, the removal of an edge takes $O(\min\{d_i(v), d_i(w)\} + h_i(v) + h_i(w))$ because we need to remove $vw$ from the $h$-graph data structure.

Finally, the diamond-free data structure can be used to test whether a graph $G$ is diamond-free, just by iteratively inserting the vertices of $G$ into the data structure.  At each step, the operations of greater complexity are the insertion of the new vertex into the $h$-graph data structure, and the computation of $N'(v)$.  As we have discussed in Section~\ref{sec:data structure}, the total time cost for these operations is $O(\alpha(G)m)$.  Thus, this algorithm runs as fast as the improvement of the algorithm by Kloks et al.\ discussed before.

\subsection{Simple, Simplicial, and Dominated Vertices of Dynamic Graphs}

A vertex $v$ is \emph{dominated} by a vertex $w$ if $N[v] \subseteq N[w]$.  Equivalently, $v$ is dominated by $w \in N(v)$ if $d(v) - d(vw) = 1$.  We say that $v$ and $w$ are \emph{comparable} if either $v$ is dominated by $w$ or $w$ is dominated by $v$.  If $v$ is dominated by all its neighbors, then $v$ is a \emph{simplicial} vertex, while if $v$ is a simplicial vertex and every pair of neighbors are comparable, then $v$ is a \emph{simple} vertex.

In~\cite{KloksKratschMullerIPL2000}, Kloks et al.\ showed how to compute the set of simplicial vertices in $O(n+m^{2\omega/(\omega+1)}) = O(n+m^{1.41})$ time, using the fast matrix multiplication algorithm.  With the $h$-graph data structure, we can find all the simple, simplicial, and dominated vertices in $O(\alpha(G) m)$ time, as follows.  First, find the degree $d(vw)$ for every $vw \in E(G)$ in $O(n+\alpha(G) m)$ time, as discussed in Section~\ref{sec:data structure}.  Next, for each vertex $v$, find the set of vertices $D(v)$ that dominate $v$, by testing whether $d(v) -d(vw) = 1$, for every $w \in N(v)$.  Clearly, if $D(v) \neq \emptyset$, then $v$ is a dominated vertex, while if $|D(v)| = d(v)$, then $v$ is a simplicial vertex.  To determine if a simplicial vertex $v$ is simple, it is enough to check whether $z \in D(w)$, for each edge $wz \in N'(v)$ with $d(w) \leq d(z)$. As discussed in Section~\ref{sec:data structure}, we can traverse all the edge-neighborhoods $N'(v)$ in $O(\alpha(G) m)$ time, thus simple, simplicial, and dominated vertices can be found in $O(\alpha(G) m)$ time.

In this section, we show how can the $h$-graph data structure be used to maintain the simple, simplicial, and dominated vertices, while vertices are inserted to or removed from a dynamic graph.  Let $G$ be a graph implemented with the $h$-graph data structure, where $d(wz)$ is stored in $\mathtt{data}(w,z)$ for every edge $wz$, and let $D$ be the family of dominated vertices.  Suppose that a new vertex $v$ with neighborhood $N(v) \subseteq V(G)$ is to be inserted into $G$, and that we want to update $D$ so as to store the dominated vertices of $G \cup \{v\}$.  The next lemma shows how to find the new dominated vertices.

\begin{lemma}
 A vertex $w \neq v$ is dominated in $H = G \cup \{v\}$ if and only if one of the following statements is true:
 \begin{itemize}
  \item $w \not \in N(v)$ and $w$ is dominated in $G$, 
  \item $d_H(w) - d_H(vw) = 1$, or
  \item $d_H(w) - d_H(wz) = 1$, for some $wz \in N_H'(v)$.
 \end{itemize}
\end{lemma}

The first step is to insert $v$ into $G$ and to compute $d(vw)$ for every $w \in N(v)$.  Both steps take $O(d(v)h(G))$ time, as discussed in Section~\ref{sec:data structure} and Theorem~\ref{thm:4-subgraph v}.  Next, we update the set $D$.  By the lemma above, we need not consider the vertices outside $N(v)$.  To find those $w \in N(v)$ that are dominated by $v$, we traverse $N(v)$ while checking whether $d(w) - d(vw) = 1$.  Next, we find all the other dominated vertices, by checking the values of $d(w) - d(wz)$ and $d(z) - d(wz)$, for every $wz \in N'(v)$.  Finally, we remove from $D$ those neighbors of $v$ that are not longer dominated, and we insert $v$ if there is some $w$ such that $d(v) - d(vw) = 1$.  Since $N'(v)$ is computed in $O(d(v)h(G))$ time, the whole procedure takes $O(d(v)h(G))$ time.

A similar procedure can be used to update the family $S$ of simplicial vertices of $G$ when $v$ is inserted into $G$.  In this case, the simplicial vertices are found as in the next lemma.

\begin{lemma}
 A vertex $w \neq v$ is simplicial in $H = G \cup \{v\}$ if and only if one of the following statements is true:
 \begin{itemize}
  \item $w \not \in N(v)$ and $w$ is simplicial in $G$, or
  \item $w$ is simplicial in $G$, and $d(w) - d(vw) = 1$.
 \end{itemize}
\end{lemma}

Again, begin with the insertion of $v$ into $G$ and the computation of $d(vw)$ for every $w \in N(v)$, in $O(d(v)h(G))$ time.  In this case, we update $S$ by first traversing each $w \in N(v)$ and checking whether $w \in S$ and $d(w) - d(vw) = 1$.  On the other hand, we insert $v$ into $S$ if and only if $d(v) - d(vw) = 1$ for every $w \in N(v)$.  The time required by these operations is $O(d(v))$, once $v$ was inserted into $G$.  Thus, the update of $S$ takes $O(d(v)h(G))$ time.

Updating the family $Q$ of simple vertices is not as simple as updating the families $D$ and $S$.  The reason is that we can no longer skip the vertices outside $N(v)$.  To provide an efficient update of $Q$, we store in $\texttt{data}(x)$, for each vertex $x$, the number $\mu(x)$ of edges $wz \in N'(x)$ such that $w$ and $z$ are not comparable.  So, $x$ is simple if and only if $x$ is simplicial and $\mu(x) = 0$.  We can find out the value of $\mu(v)$ in $O(d(v)h(G))$ time, by traversing every $wz \in N'(v)$ and checking whether $\min\{d(w),d(z)\} - d(wz) = 1$.  The update of all the other values of $\mu$ is based on the following lemma.

\begin{lemma}
 Two vertices $w$ and $z$ of $G$ are comparable in $G \cup \{v\}$ if and only if they are comparable in $G$ and either $\{w,z\} \subseteq N(v)$ or $\{w, z\} \cap N(v) = \emptyset$.
\end{lemma}

The update of $\mu(x)$ for every $x \neq v$ is done as in Algorithm~\ref{alg:mu update}.
\begin{algorithm} \label{alg:mu update}
 Update of $\mu$ after the insertion of $v$.
 \begin{enumerate}
 \item For each $w \in N(v)$:
 \item\hspace*{5mm} For each $z \in H(w)$:
 \item\hspace*{10mm} If $z \not \in N(v)$, $d_G(w) - d_G(wz) = 1$ and $d_G(z) - d_G(wz) > 1$, then:
 \item\hspace*{10mm} \{$w$ and $z$ were comparable before inserting $v$, but now they are not\}
 \item\hspace*{15mm} Set $\mu(x) = \mu(x)+1$, for every $x \in N(wz)$.
\end{enumerate}
\end{algorithm}
Algorithm~\ref{alg:mu update} requires $O(m)$ time in the worst case.  However, suppose that Algorithm~\ref{alg:mu update} is iteratively executed for computing the simple vertices as in Algorithm~\ref{alg:simple insertion}.  In such an algorithm, the condition of the inner loop of Algorithm~\ref{alg:mu update} is executed at most once for each edge $wz$.  Indeed, if $w$ and $z$ are not comparable prior the insertion of $v$, then they are not comparable after the insertion of $v$.  On the other hand, when $wz$ meets the condition of the inner loop, we know that $d_G(w) \leq d_G(z)$.  Hence, as discussed in Section~\ref{sec:data structure}, the time required by Algorithm~\ref{alg:simple insertion} is $O(n+\alpha(G)m)$ for the update of the $h$-graph data structure and $d(vw)$, and \[O\left(n + \sum_{wz \in E(G)}{\min\{d(v), d(z)\}}\right) = O(n+\alpha(G)m)\] for the update of $\mu$.
\begin{algorithm} \label{alg:simple insertion}
 Iterative update of $\mu$ for a graph $G$.
 \begin{enumerate}
  \item Let $v_1, \ldots, v_n$ be an ordering of $V(G)$, and $G'$ be an empty graph.
  \item For $i = 1, \ldots, n$: insert $v_i$ into $G'$ while executing Algorithm~\ref{alg:mu update}.
 \end{enumerate}
\end{algorithm}

As a corollary, we obtain that the simple, simplicial, and dominated vertices of $G$ can be found in $O(n+\alpha(G)m)$ time.  The process for updating the sets of simple, simplicial, and dominated vertices when a vertex is removed is similar.  When $v$ is removed, again we should not consider those vertices outside $N(v)$ for the update of $D$ and $S$.  A vertex in $N(v)$ is dominated in $G \setminus \{v\}$ if, in $G \setminus \{v\}$, $d(w) - d(wz) = 1$ for some $z \in H(w)$, while it is simplicial in $G \setminus \{v\}$ if $L(w) = \emptyset$ and $d(w) - d(wz) = 1$ for every $z \in H(w)$.  Thus, we can update $D$ and $S$ in $O(d(v)h(G))$ time.  Finally, we can update $\mu$ in $O(m)$ time when $v$ is removed by applying Algorithm~\ref{alg:mu remove}, which is the inverse of Algorithm~\ref{alg:mu update}.
\begin{algorithm} \label{alg:mu remove}
 Update of $\mu$ after the insertion of $v$.
 \begin{enumerate}
 \item For each $w \in N(v)$:
 \item\hspace*{5mm} For each $z \in H(w)$:
 \item\hspace*{10mm} If $z \not \in N(v)$, $d_G(w) - d_G(wz) = 2$ and $d_G(z) - d_G(wz) > 1$, then:
 \item\hspace*{10mm} \{$w$ and $z$ were not comparable before inserting $v$, but now they are\}
 \item\hspace*{15mm} Set $\mu(x) = \mu(x)-1$, for every $x \in N(wz)$.
\end{enumerate}
\end{algorithm}

The removal operations can be used to improve the best known algorithm for the recognition of cop-win and strongly perfect graphs, for sparse graphs.

A \emph{cop-win order} of a graph $G$ is an ordering $v_1, \ldots, v_n$ of $V(G)$ such that $v_i$ is dominated in the subgraph induced by $v_i, \ldots, v_n$, for $1 \leq i \leq n$.  A graph that admits a cop-win order is a \emph{cop-win graph}.  The cop-win name comes from the fact that cop-win graphs are precisely the graphs in which a cop can always catch the robber in a pursuit game~\cite{NowakowskiWinklerDM1983}.  This class has been introduced in~\cite{Poston1971}, cf.~\cite{BandeltPrisnerJCTSB1991}.  Cop-win graphs are also known in the literature under the name of dismantlable graphs~\cite{Quilliot1983}, and they are a main tool in the study of clique graphs~\cite{Szwarcfiter2003}.  The currently best algorithms for recognizing cop-win graphs run in $O(nm)$ time or in $O(n^3/\log n)$ time~\cite{SpinradDAM2004}.

A \emph{dismantling} of a graph $G$ is a graph $H$ obtained by iteratively removing one dominated vertex of $G$, until no more dominated vertices remain. It is not hard to see that all the dismantlings of $G$ are isomorphic.  Using the $h$-graph data structure, we can compute the dismantling of a graph, and the cop-win order of a cop-win graph, in $O(\alpha(G)m)$ time easily.  First, find the set $D$ of dominated vertices in $O(n+\alpha(G)m)$ time (cf.\ above).  Then, while $D\neq \emptyset$, choose a vertex $D$, and remove it from $G$ while updating $D$ as explained above.  The graph obtained after this procedure is the dismantling $H$ of $G$.  If $H$ has a unique vertex, then $G$ has a cop-win order, given by the order in which the vertices were removed by the algorithm.  This algorithm takes $O(n+\alpha(G)m)$ time, as discussed above.

A \emph{simple elimination ordering} of a graph $G$ is an ordering $v_1, \ldots, v_n$ of $V(G)$ such that $v_i$ is a simple vertex in the subgraph of $G$ induced by $v_i, \ldots, v_n$.  The family of graphs that admit a simple elimination ordering is precisely the family of strongly chordal graphs~\cite{FarberDM1983}.  Strongly chordal graphs were introduced as a subclass of chordal graph for which the domination problem is solvable in polynomial time~\cite{FarberDM1983}.  The best known algorithms for computing a simple elimination ordering run in either $O(n^2)$ or $O((n+m)\log n)$ time~\cite{Spinrad2003}.  These algorithms work by finding a doubly lexical ordering of the adjacency matrix of the graph, and then testing if this sorted matrix contains some forbidden structure.  Our approach, instead, is based on iteratively finding a simple elimination ordering by iteratively removing the simple vertices.  

Every strongly chordal graph $G$ has at least one simple vertex, and $G \setminus \{v\}$ is strongly chordal for every $v \in V(G)$~\cite{FarberDM1983}.  Thus, we can compute a simple elimination ordering of $G$ by removing the simple vertices in any order.  That is, first find the set $Q$ of simple vertices in $O(n+\alpha(G)m)$ time (cf.\ above).  Then, while $Q \neq \emptyset$, choose a vertex in $Q$, and remove it from $G$ while updating $Q$ with Algorithm~\ref{alg:mu remove}.  If all the vertices are removed from $Q$, then the removal order given by the algorithm is a simple elimination ordering of $G$.  As discussed for Algorithm~\ref{alg:simple insertion}, the inner loop of Algorithm~\ref{alg:mu remove} is executed at most once for each edge (cf.\ above).  Therefore, the simple elimination ordering is computed in $O(n+\alpha(G)m)$ time, improving the previous best algorithms for graphs with low arboricity.

\end{document}